\documentclass[preprint,12pt]{elsarticle}
\journal{Systems and Control Letters}

\usepackage{amsmath,amsthm,amssymb,amsfonts}
\usepackage{graphicx}
\usepackage{subfig}
\usepackage{cuted}
\usepackage{textcomp}
\usepackage{color}
\usepackage{xcolor}
\usepackage{multirow}
\usepackage{multicol,lipsum}
\usepackage{epsfig}
\usepackage{algorithm}
\usepackage{algorithmic}
\usepackage{multirow}
\usepackage{hyperref}
\usepackage{mathtools}
\usepackage{bm}
\usepackage{listings}
\newtheorem{remark}{Remark}
\newtheorem{definition}{Definition}
\newtheorem{theorem}{Theorem}

\newtheorem{proposition}{Proposition}
\newtheorem{lemma}{Lemma}
\newtheorem{example}{Example}

\newtheorem{assumption}{Assumption}
\usepackage{array}
\usepackage{float}
\usepackage[short]{optidef}
\usepackage{mdframed}
\mdfdefinestyle{mystyle}{
    backgroundcolor=white!10}

\newcommand{\oomit}[1]{}
\usepackage{graphicx}          

\begin{document}

\begin{frontmatter}

\title{Converse Barrier Certificates for Finite-time Safety Verification of Continuous-time Perturbed Deterministic Systems\footnote{To appear in Systems \& Control Letters}} 


\author{Yonghan Li$^{1,2}$, Chenyu Wu$^{1,2}$, Taoran Wu$^{1,2}$, Shijie Wang$^3$ and Bai Xue$^{1,2}$}   

\address{1. Key Lab. of System Software (Chinese Academy of Sciences), Institute of Software, CAS, Beijing, China \\ \{wucy,wutr,xuebai\}@ios.ac.cn} 

\address{2. University of Chinese Academy of Sciences, Beijing, China}

\address{3. School of Mathematics and Physics, University of Science and Technology Beijing, Beijing, China \\
wangshijie@ustb.edu.cn}  
          

\begin{abstract}                          
In this paper, we investigate the problem of verifying the finite-time safety of continuous-time perturbed deterministic systems represented by ordinary differential equations in the presence of measurable disturbances. Given a finite-time horizon, if the system is safe, it, starting from a compact initial set, will remain within an open and bounded safe region throughout the specified time horizon, regardless of the disturbances. The main contribution of this work is a converse theorem: we prove that a continuously differentiable, time-dependent barrier certificate exists if and only if the system is safe over the finite-time horizon. The existence problem is explored by finding a continuously differentiable approximation of a unique Lipschitz viscosity solution to a Hamilton-Jacobi equation.
\end{abstract}

\end{frontmatter}

\section{Introduction}
The rapid growth of modern technology has resulted in the proliferation of intelligent autonomous systems. Ensuring the safe operation of these systems is of utmost importance, but the presence of external disturbances introduces uncertainties that must be addressed to guarantee the system's safety. Establishing set invariance as a means of ensuring safety has proven to be a valuable approach in a range of robotics and control applications. In recent years, the utilization of barrier certificates to guarantee set invariance and safety in control dynamical systems has gained widespread recognition, particularly in safety-critical contexts \cite{ames2019control}.

The barrier certificates method, first introduced in \cite{prajna2004safety,prajna2007framework}, was initially proposed for safety verification of continuous-time deterministic and stochastic systems over the infinite-time horizon. Over time, several barrier certificates for infinite-time safety verification have been developed, each offering different levels of expression and capabilities. Meanwhile, converse theorems for barrier certificates, which concern the existence of barrier certificates, have played a crucial role in understanding how safety properties can be characterized by barrier certificates. As a result, they have gained increasing attention since the introduction of barrier certificates, and have been explored in \cite{prajna2005necessity,prajna2007convex,wisniewski2015converse,ratschan2018converse,liu2021converse,maghenem2022converse,10316551}. However, the requirement for safety verification over an infinite-time horizon, particularly under external and/or internal perturbations, can be considered overly stringent—especially in contexts where practical applications operate under defined time constraints and can demonstrate consistent safety within those boundaries. For example, many real-world systems, such as those in automotive or aerospace engineering, function within specific operational periods that necessitate finite-time safety verification, making it a more relevant and achievable metric for assessing safety. Despite this, there has been relatively less focus on applying barrier certificates to finite-time safety verification for continuous-time deterministic systems \cite{gan2016barrier}. In contrast, there is growing interest in finite-time safety verification for continuous-time stochastic systems \cite{steinhardt2012finite,santoyo2021barrier,wang2021safety,xue2023new}. Additionally, the converse theorem of barrier certificates for finite-time safety verification has yet to be thoroughly investigated.

In this paper, we investigate the use of time-dependent barrier certificates for finite-time safety verification of continuous-time systems governed by ordinary differential equations with perturbation inputs. Finite-time safety verification aims to establish whether, starting from a predetermined initial set, the system evolves within a designated safe set throughout the given time horizon. The main contribution of this study is the revelation that the existence of a continuously differentiable time-dependent barrier certificate is both a necessary and sufficient condition for the system's safety over finite-time horizons. This revelation is explored within the framework of Hamilton-Jacobi reachability. Initially, we design a value function that remains strictly negative over the initial set if and only if the system is safe within the finite-time horizon. This value function is defined based on the function that describes the safe set and the system's trajectories. Subsequently, we reduce the value function to the unique Lipschitz viscosity solution of a Hamilton-Jacobi partial differential equation. As a third step, we derive a condition from the Hamilton-Jacobi equation and demonstrate that this condition allows for the existence of a continuously differentiable function if and only if the system is safe over the finite-time horizon. Finally, we show that this continuously differentiable function serves as a time-dependent barrier certificate.

\subsection{Related Work}
Hamilton-Jacobi equations have been widely utilized in solving optimal control problems \cite{bardi1997optimal} and, in recent decades, have been extended to reachability analysis (e.g.,  \cite{mitchell2005time,margellos2011hamilton,xue2022differential}). In this framework, reachable sets correspond to specific (sub)super level sets of the equation's viscosity solutions. While substantial work has focused on scalable numerical approximations to mitigate the computational challenge of finding exact solutions \cite{bansal2017hamilton}, alternative approaches have sought efficiency through relaxation. Notably, \cite{xue2019robust,xue2019inner} introduced barrier-like constraints for computing inner approximations of robust invariant sets and reach-avoid sets. These constraints are obtained via relaxing Hamilton-Jacobi equations, and can be efficiently solved using semi-definite programming tools, particularly for polynomial systems. Notably, \cite{xue2019inner} demonstrated the existence of a continuously differentiable solution to the established constraint, which can approximate the viscosity solution of the original Hamilton-Jacobi equation with arbitrary accuracy. In this paper, we extend the methodology proposed by \cite{xue2019inner} to the finite-time safety verification of continuous-time deterministic systems and establish that there exists a time-dependent barrier certificate if and only if the system is safe. We emphasize that although both barrier-certificate computation and Hamilton–Jacobi reachability analysis face computational challenges for general nonlinear systems, our theoretical results establish the fundamental completeness of the barrier-certificate approach for finite-time safety verification. This existence guarantee allows computational efforts to focus on developing efficient approximation methods without fundamental limitations. While our proof builds upon the approximation framework in \cite{xue2019inner}, it extends that framework by establishing necessary and sufficient conditions specifically for finite-time safety through barrier certificates. This theoretical foundation motivates the design of specialized computational methods that may alleviate the curse of dimensionality inherent in direct Hamilton–Jacobi solutions.

Converse theorems for barrier certificates have received attention since their introduction. \cite{prajna2005necessity,prajna2007convex,wisniewski2015converse,ratschan2018converse,liu2021converse,10316551} investigated the existence of time-independent barrier functions for the infinite-time safety verification of (perturbed) deterministic systems under certain assumptions. Under mild regularity conditions, \cite{maghenem2022converse} investigated converse theorems for safety over the infinite-time horizon using time-dependent barrier functions. For systems dynamics described by a set-valued map, it explored the existence of lower semicontinuous and locally Lipschitz time-dependent barrier functions. In contrast, the present work studied converse theorems for safety over finite-time horizons in terms of continuously differentiable, time-dependent barrier functions. Moreover, although finite-time safety verification is considered, the time horizon can be arbitrarily large but bounded.

In addition, numerous well-established reachability methods \cite{althoff2021set} have been developed to compute over-approximations of reachable states for finite-time safety verification. Many methods achieve convergence in the Hausdorff distance (e.g., \cite{rungger2018accurate}), providing rigorous set approximations that can avoid unsafe regions. However, some approaches rely on measure-theoretic convergence (e.g., \cite{henrion2013convex}), which may not theoretically ensure that over-approximations completely avoid intersecting unsafe sets, even if the system is safe. Our work complements Hausdorff-convergent reachability methods by offering an alternative verification paradigm through barrier certificates. Specifically, we establish that a continuously differentiable, time-dependent barrier certificate exists if and only if the system is safe over finite-time horizons, providing a function-based safety proof without explicit reach set computation. This contrasts with reachability analysis that focuses on set computations, and each methodology has distinct advantages depending on the verification context.
\color{black}

\section{Preliminaries}
\label{pre}
In this section, we present the concepts of continuous-time perturbed deterministic systems, finite-time safety verification problem, and time-dependent barrier certificates for addressing it. Before posing them, let us introduce some basic notions used throughout this paper: $\mathbb{R}$ stands for the set of reals. $\mathbb{R}[\boldsymbol{x}]$ denotes the ring of polynomials in the variables $\boldsymbol{x}$ with real coefficients. The closure of a set $X$ is denoted by $\overline{X}$, its complement by $X^c$ and its boundary by $\partial X$.  Vectors are denoted by boldface letters.  

The continuous-time perturbed deterministic system considered in this work has dynamics described by the following ordinary differential equation subject to disturbances:
\begin{equation}
\label{sys}
\dot{\bm{x}} = \bm{f}(\bm{x},\bm{d}) = \bm{f}_1(\bm{x}) + \bm{f}_2(\bm{x}) \bm{d}(t), \quad
\bm{x}(0) = \bm{x}_0 \in \mathbb{R}^n,
\end{equation}
where $\dot{\bm{x}} = \frac{d\bm{x}(t)}{dt}$ denotes the time derivative of the state $\bm{x}(t)$. Here, 
\[
\bm{f}_1(\bm{x}) = (f_{1,1}(\bm{x}), \ldots, f_{1,n}(\bm{x}))^\top, \quad
\bm{f}_2(\bm{x}) = (f_{2,1}(\bm{x}), \ldots, f_{2,n}(\bm{x}))^\top,
\]
with $f_{1,i}(\cdot) \colon \mathbb{R}^n \to \mathbb{R}$ and $f_{2,i}(\cdot)\colon \mathbb{R}^n \to \mathbb{R}^m$ for $i=1,\ldots,n$. The vector $\bm{d}(t) = (d_1(t), \ldots, d_m(t))^\top$ represents a collection of disturbance inputs, with each $d_i(\cdot) \colon [0,\infty) \to \mathbb{R}$ for $i=1,\ldots,m$.

 The following assumptions are imposed on the system \eqref{sys}. 
\begin{assumption}
\label{assmp}
  1) The disturbance input $\bm{d}(\cdot)\colon \mathbb{R}\rightarrow \mathcal{D}$ is a measurable function, where $\mathcal{D}\subseteq \mathbb{R}^m$ is compact and convex. Let  $\mathcal{M}$ denote the set of measurable functions from $[0,\infty)$ to $\mathcal{D}$, i.e., 
        $\mathcal{M}: =\{\bm{d}(\cdot)\colon [0,\infty)\rightarrow \mathcal{D} \text{~is~measurable}\}$.
   2) The function $\bm{f}$ is locally Lipschitz continuous with respect to $\bm{x}$ uniformly over $\bm{d}\in \mathcal{D}$, i.e., for each compact set $\mathcal{X}\subseteq \mathbb{R}^n$, there exists $L_{\bm{f}}$ such that 
        \[\|\bm{f}(\bm{x},\bm{d})-\bm{f}(\bm{y},\bm{d})\|\leq L_{\bm{f}}\|\bm{x}-\bm{y}\|, \forall \bm{x}, \bm{y}\in \mathcal{X}, \forall \bm{d}\in \mathcal{D}.\]
\end{assumption}

\begin{remark}
Although the finite-time safety verification problem considered in this paper concerns only system trajectories over a finite horizon $[0,T]$, the disturbance inputs are required to be defined on $[0,\infty)$ for a specific reason. 
The Hamilton-Jacobi equation used in our analysis in Section~\ref{problem_solving} is posed on an infinite time domain and requires the existence of global-in-time solutions to the associated partial differential equation. 
To ensure well-posedness of this Hamilton--Jacobi equation, it is therefore standard to consider disturbance signals $\bm{d}(\cdot)$ defined on $[0,\infty)$, even though only their restriction to $[0,T]$ is relevant for finite-time safety verification.
\end{remark}

Given a measurable input $d(\cdot)\in\mathcal M$ and an initial state $\bm{x}_0\in\mathbb{R}^n$,
the induced trajectory of system \eqref{sys} is denoted by
$\bm{\phi}^{\bm{d}}_{\bm{x}_0}(\cdot):[0,\tau_{\bm{x}_0})\to\mathbb{R}^n$,
where $[0,\tau_{\bm{x}_0})$ is the maximal interval of existence of the solution
with initial condition $\bm{\phi}^{\bm{d}}_{\bm{x}_0}(0)=\bm{x}_0$. Here, $\tau_{\bm{x}_0}\in(0,\infty]$ denotes the supremum of all times
for which the solution exists.

We now introduce the finite-time safety verification problem. Given an open and bounded safe set $\mathcal{S}\subseteq \mathbb{R}^n$ and an initial set $\mathcal{X}_0\subseteq \mathcal{S}$ which is compact, the finite-time safety verification problem is to determine whether the system \eqref{sys}, starting from $\mathcal{X}_0$, will evolve within $\mathcal{S}$  over the time horizon $[0,T]$ with $T<\infty$, where 
\begin{equation*}
\label{set}
\begin{split}
\mathcal{S}=\{\bm{x}\in \mathbb{R}^n\mid h(\bm{x})< 0\} \text{~and~} \partial \mathcal{S}=\{\bm{x}\in \mathbb{R}^n\mid h(\bm{x})=0\}
\end{split}
\end{equation*}
with $h(\cdot)\colon \mathbb{R}^n\rightarrow \mathbb{R}$ being locally Lipschitz.

\begin{definition}[Finite-time Safety Verification]\label{RNS}
Given system \eqref{sys}, a bounded and open safe set $\mathcal{S}\subseteq \mathbb{R}^n$, a compact initial set $\mathcal{X}_0\subseteq \mathcal{S}$, and a time interval $[0,T]$ with $T<\infty$, the system \eqref{sys} is safe over the time horizon $[0,T]$, if starting from the initial set $\mathcal{X}_0$, it will evolve within the safe set $\mathcal{S}$ over the time horizon $[0,T]$, regardless of disturbances $\bm{d}(\cdot)\in \mathcal{M}$, i.e.,   
\[ \forall \bm{x}_0\in \mathcal{X}_0. \forall \bm{d}\in \mathcal{M}. \forall t\in [0,T].  \bm{\phi}_{\bm{x}_0}^{\bm{d}}(t) \in \mathcal{S}.\]
The finite-time safety verification problem is to verify whether the system \eqref{sys} is safe over the time horizon $[0,T]$.
\end{definition}

The finite-time safety verification problem can be addressed via a time-dependent barrier certificate $v(\bm{x},t)$.
\begin{definition}[Time-Dependent Barrier Certificates]
A continuously differentiable function $v(\cdot,\cdot)\colon \mathbb{R}^n \times [0,T]\rightarrow \mathbb{R}$ is said to be a time-dependent barrier certificate if the following condition is satisfied:
\begin{equation}
\label{Finite_safe}
    \begin{cases}
        v(\bm{x},0)< 0, & \forall \bm{x}\in \mathcal{X}_0,\\
        v(\bm{x},t)\geq 0,    & \forall (\bm{x},t)\in \partial \mathcal{S}\times [0,T],\\
        \mathcal{L}^{\bm{d}}v(\bm{x},t)\leq 0, & \forall (\bm{x},t) \in \overline{\mathcal{S}}\times [0,T], \ \forall \bm{d}\in \mathcal{D},
    \end{cases}
    \end{equation}
  where, for a given $\bm d\in\mathcal D$, $\mathcal{L}^{\bm{d}}v(\bm{x},t):= \frac{\partial v(\bm{x},t)}{\partial t}
+ \frac{\partial v(\bm{x},t)}{\partial \bm{x}}\bm{f}(\bm{x},\bm{d})$.
\end{definition}

\begin{proposition}
\label{prop1}
If there exists a  time-dependent barrier certificate $v$ satisfying \eqref{Finite_safe}, the system \eqref{sys} is safe over the time horizon $[0,T]$.
\end{proposition}

The proof of Proposition \ref{prop1} is shown in Appendix A. However, the converse problem remains unanswered. That is, if the system \eqref{sys} is safe over the finite time horizon $[0,T]$, is there a time-dependent barrier certificate $v$ that satisfies \eqref{Finite_safe}? In this paper, we will give a positive answer to this question.

\section{Converse Theorem for Safety Verification}
\label{problem_solving}
In this section, we present our methodology to demonstrate the existence of a time-dependent barrier certificate if the system \eqref{sys} is safe over the time horizon $[0,T]$. Our approach is rooted in the framework of Hamilton-Jacobi reachability, which investigates reachability analysis through the use of a Hamilton-Jacobi equation.

The construction of the Hamilton-Jacobi equation follows the one in \cite{xue2019inner}. Since the function $\bm{f}$ is locally Lipschitz with respect to $\bm{x}$, uniformly over $\bm{d}\in \mathcal{D}$, we first construct an auxiliary vector field $\bm{F}(\cdot,\cdot)\colon \mathbb{R}^n\times \mathcal{D}\rightarrow \mathbb{R}^n$ as in \cite{xue2019inner,xue2019robust}, which is required to satisfy
\[\bm{f}(\bm{x},\bm{d})=\bm{F}(\bm{x},\bm{d})=\bm{f}_1(\bm{x})+\bm{f}_2(\bm{x})\bm{d}, \forall \bm{x}\in \overline{\mathcal{S}}, \forall \bm{d}\in \mathcal{D}.\]
Additionally, it is globally Lipschitz in $\bm{x}$, uniformly over $\bm{d}\in \mathcal{D}$, i.e., there exists a constant $L_{F}$ such that for $\bm{x}_1,\bm{x}_2\in \mathbb{R}^n$, 
\[\|\bm{F}(\bm{x}_1,\bm{d})-\bm{F}(\bm{x}_2,\bm{d})\|\leq L_{F}\|\bm{x}_1-\bm{x}_2\|, \forall \bm{d}\in \mathcal{D}.\]
This Lipschitz condition implies that there exists a constant $C>0$ such that $\|\bm{F}(\bm{x},\bm{d})\|\leq C(1+\|\bm{x}\|)$. This is obtained as follows:  since $\|\bm{F}(\bm{x},\bm{d})-\bm{F}(\bm{0},\bm{d})\|\leq L_{F}\|\bm{x}\|, \forall \bm{d}\in \mathcal{D}$, $\|\bm{F}(\bm{x},\bm{d})\|\leq \max\{L_F, \max_{\bm{d}\in \mathcal{D}}\bm{F}(\bm{0},\bm{d})\} (1+\|\bm{x}\|)$ holds and thus the conclusion holds. For this, we construct $\bm{F}$ by \emph{separately extending} $\bm{f}_1$ and $\bm{f}_2$ from $\overline{\mathcal{S}}$ to $\mathbb{R}^n$ using standard Lipschitz extension techniques (e.g., Kirszbraun’s theorem) and then defining
\[
\bm{F}(\bm{x},\bm{d}) := \bm{F}_1(\bm{x}) + \bm{F}_2(\bm{x})\bm{d}.
\]  
This construction guarantees that $\bm{F}$ is globally Lipschitz in $\bm{x}$ while preserving the affine structure in $\bm{d}$. This extension is possible because the safe set $\mathcal{S}$ is open and bounded, 
and hence its closure $\overline{\mathcal{S}}$ is compact. A formal statement of this result is provided in Proposition~\ref{exl} in Appendix~B. Therefore, we assume the availability of this function by default.

Given a measurable input $\bm{d}(\cdot)\in \mathcal{M}$ and an initial state $\bm{x}_0\in \mathbb{R}^n$, the induced trajectory of the system $\dot{\bm{x}}=\bm{F}(\bm{x},\bm{d})$ is denoted by $\bm{\psi}_{\bm{x}_0}^{\bm{d}}(\cdot): [0,\infty)\rightarrow \mathbb{R}^n$ with $\bm{\psi}_{\bm{x}_0}^{\bm{d}}(0)=\bm{x}_0$.
\begin{lemma}
 Given the safe set $\mathcal{S}$, initial set $\mathcal{X}_0$, and finite-time horizon $[0,T]$, the system \eqref{sys} is safe over the time horizon $[0,T]$ if and only if the system $\dot{\bm{x}}=\bm{F}(\bm{x},\bm{d})$ is safe over the time horizon $[0,T]$. 
\end{lemma}
\begin{proof}
Given $\bm{d}(\cdot)\in \mathcal{M}$ and $\bm{x}_0\in \mathcal{X}_0$, we denote $\tau=\inf\{t\mid \bm{\phi}_{\bm{x}_0}^{\bm{d}}(t)\in \partial \mathcal{S} \wedge t\leq T\}$, where $\bm{x}_0\in \mathcal{X}_0$. Since $\bm{f}(\bm{x},\bm{d})=\bm{F}(\bm{x},\bm{d}), \forall \bm{x}\in \overline{\mathcal{S}}, \forall \bm{d}\in \mathcal{D}$, we obtain $\bm{\phi}_{\bm{x}_0}^{\bm{d}}(t)=\bm{\psi}_{\bm{x}_0}^{\bm{d}}(t)$ for $t\in [0,\tau]$. Thus, we have the conclusion.  
\end{proof}
Additionally, since $\bm{f}(\bm{x},\bm{d})=\bm{F}(\bm{x},\bm{d}), \forall \bm{x}\in \overline{\mathcal{S}}, \forall \bm{d}\in \mathcal{D}$, we have that a function $v(\cdot,\cdot)\colon \mathbb{R}^n \times [0,T]$ is a time-dependent barrier certificate for the system \eqref{sys} if and only if it is a time-dependent barrier certificate  for the system $\dot{\bm{x}}=\bm{F}(\bm{x},\bm{d})$. Thus, we use the system 
\begin{equation}
\label{sys1}
  \dot{\bm{x}}=\bm{F}(\bm{x},\bm{d}), \forall \bm{d}\in \mathcal{D}
\end{equation}
to show the existence of a time-dependent barrier certificate if the system \eqref{sys} is
safe over $[0,T]$. 

Our method begins with a value function defined for all  $\bm{x}\in \mathbb{R}^n$ and $t\in [0,T]$ as:
 \begin{equation}
 \label{value}
     V(\bm{x},t):=\sup_{\bm{d}\in \mathcal{M}}\sup_{\tau \in [t,T]}h(\bm{\psi}^{\bm{d}}_{\bm{x}}(\tau-t)). 
 \end{equation}  
\begin{lemma}
\label{continuous}
The value function $V$ in \eqref{value}  is locally Lipschitz continuous on $\mathbb{R}^n \times [0,T]$.
\end{lemma}
\begin{proof}
The conclusion can be justified via following the proof of Lemma 1 in \cite{xue2019inner}. Its proof is also presented in Appendix A.
\end{proof}
 \begin{lemma}
 \label{equiv}
The system \eqref{sys1}, starting from the state $\bm{x}_0 \in \mathcal{X}_0$ at $t=0$, is safe over the time horizon $[0,T]$ if and only if $V(\bm{x}_0,0)<0$.
 \end{lemma}
 \begin{proof}
  Given an initial state $\bm{x}_0\in \mathcal{X}_0$, if the system \eqref{sys1} is safe over the time horizon $[0,T]$, we have $\bm{\psi}_{\bm{x}_0}^{\bm{d}}(t) \in \mathcal{S}, \forall (t,\bm{d})\in [0,T] \times \mathcal{M}$.
This implies $h(\bm{\psi}_{\bm{x}_0}^{\bm{d}}(t))<0, \forall (t,\bm{d})\in  [0,T] \times \mathcal{M}$. Let $M=V(\bm{x}_0,0)$. Since $\mathcal{D}$ is compact and the convex hull of the set $\bm{F}(\bm{x},\mathcal{D})$ is equal to  $\bm{F}(\bm{x},\mathcal{D})$ for $\bm{x}\in \overline{\mathcal{S}}$, the set of trajectories for the system \eqref{sys1} starting from $\bm{x}_0$ is closed with respect to the uniform convergence over a bounded interval \cite{bardi1997optimal}. Consequently, there exists $\bm{d}_1\in \mathcal{M}$ such that $M=\sup_{t\in [0,T]}h(\bm{\psi}^{\bm{d}_1}_{\bm{x}_0}(t))$. Furthermore, since $h(\bm{\psi}^{\bm{d}_1}_{\bm{x}_0}(t))$ is continuous on $t\in [0,T]$, we have that there exists $\tau\in [0,T]$ such that $M=h(\bm{\psi}^{\bm{d}_1}_{\bm{x}_0}(\tau))$. Thus, $M<0$, which implies $V(\bm{x}_0,0)<0$. 

On the other hand, if $V(\bm{x}_0,0)<0$, we have $\bm{\psi}_{\bm{x}_0}^{\bm{d}}(t) \in \mathcal{S}, \forall (t,\bm{d})\in [0,T] \times \mathcal{M}$ and consequently the system \eqref{sys1} is safe over the time horizon $[0,T]$.
 \end{proof}

\begin{lemma}
 \label{equiv1}
 The system \eqref{sys1}, starting from  $\mathcal{X}_0$ at $t=0$, is safe over the time horizon $[0,T]$ if and only if there exists $\delta>0$ such that  $V(\bm{x}_0,0)\leq -\delta, \forall \bm{x}_0\in \mathcal{X}_0$.
\end{lemma}
\begin{proof}
If there exists $\delta>0$ such that  $V(\bm{x}_0,0)\leq -\delta, \forall \bm{x}_0\in \mathcal{X}_0$, we can obtain that the system \eqref{sys1}, starting from  $\mathcal{X}_0$ at $t=0$, is safe over the time horizon $[0,T]$. 

On the other hand, if the system \eqref{sys1}, starting from  $\mathcal{X}_0$ at $t=0$, is safe over the time horizon $[0,T]$, we have $V(\bm{x}_0,0)<0, \forall \bm{x}_0\in \mathcal{X}_0$ according to Lemma \ref{equiv}. Furthermore, from Lemma \ref{continuous}  and the fact that $\mathcal{X}_0$ is compact, we have that there exists $\delta>0$ such that  
$V(\bm{x}_0,0)\leq -\delta, \forall \bm{x}_0\in \mathcal{X}_0$.
The proof is completed.
\end{proof}
 The value function $V$ in \eqref{value} is the unique Lipschitz viscosity solution to a Hamilton-Jacobi type equation. 
\begin{proposition}
\label{HJ}
  The value function \eqref{value} is the unique Lipschitz viscosity solution  to the following Hamilton-Jacobi equation: for $(t,\bm{x})\in [0,T]\times \mathbb{R}^n$, 
  \begin{equation}
  \label{hami}
  \begin{cases}
\max \big\{h(\bm{x})-V(\bm{x},t), \frac{\partial V(\bm{x},t)}{\partial t}+\sup_{\bm{d}\in \mathcal{D}}\frac{\partial V(\bm{x},t)}{\partial \bm{x}}\bm{F}(\bm{x},\bm{d})\big\}=0,\\
V(\bm{x},T)=h(\bm{x}).
\end{cases}
\end{equation}
\end{proposition}
\begin{proof}
    The conclusion can be obtained via following the proof of Theorem 2 in \cite{xue2019inner}. Its proof is also presented in Appendix A.
\end{proof}
Via removing the maximum  operator in Equation \eqref{hami} and using the fact that $\bm{f}(\bm{x},\bm{d})=\bm{F}(\bm{x},\bm{d}), \forall \bm{x}\in \overline{\mathcal{S}}, \forall \bm{d}\in \mathcal{D}$, we can obtain a set of constraints, which admits a continuously differentiable solution if and only if the system \eqref{sys} is safe over the time horizon $[0,T]$. 
 \begin{lemma}
 \label{converse_pre}
   The system \eqref{sys} is safe over the time horizon $[0,T]$ if and only if there exists a continuously differentiable function $v(\cdot,\cdot): \mathbb{R}^n\times [0,T]\rightarrow \mathbb{R}$ satisfying 
    \begin{equation}
    \label{sufficient}
        \begin{cases}
             v(\bm{x},0)<0, & \forall \bm{x}\in \mathcal{X}_0,\\
            h(\bm{x})-v(\bm{x},t)\leq  0, & \forall (\bm{x},t)\in \overline{\mathcal{S}} \times [0,T],\\
           \mathcal{L}^{\bm{d}} v(\bm{x},t)\leq 0, & \forall (\bm{x},t)\in \overline{\mathcal{S}} \times [0,T], \forall \bm{d}\in \mathcal{D},
        \end{cases}
    \end{equation}
    where, for a given $\bm d\in\mathcal D$, $\mathcal{L}^{\bm{d}} v(\bm{x},t):=\frac{\partial v(\bm{x},t)}{\partial t}+\frac{\partial v(\bm{x},t)}{\partial \bm{x}}\bm{f}(\bm{x},\bm{d})$.
 \end{lemma}
\begin{proof}
From the constraint $h(\bm{x})-v(\bm{x},t)\leq  0, \forall (\bm{x},t)\in \overline{\mathcal{S}} \times [0,T]$, we have 
$v(\bm{x},t)\geq 0, \forall (\bm{x},t)\in \partial \mathcal{S}\times [0,T]$. Thus, if there exists a continuously differentiable function $v$ satisfying \eqref{sufficient}, the system \eqref{sys} starting from the initial set $\mathcal{X}_0$ cannot touch the boundary  $\partial \mathcal{S}$ over $[0,T]$ and thus the system \eqref{sys} is safe over the time horizon $[0,T]$. 

On the other hand, since the system is safe over the time horizon $[0,T]$, according to Lemma \ref{equiv1}, there exists $\delta>0$ such that $\max_{\bm{x}\in \mathcal{X}_0}V(\bm{x},0)=-\delta$.  Due to $\bm{f}(\bm{x},\bm{d})=\bm{F}(\bm{x},\bm{d})$ for $\bm{x}\in \overline{\mathcal{S}}$ and $\bm{d}\in \mathcal{D}$,  and the value function $V$ in \eqref{value} is Lipschitz continuous over $\overline{\mathcal{S}}\times [0,T]$, the value function $V$ satisfies 
$\frac{\partial V(\bm{x},t)}{\partial t}+\frac{\partial V(\bm{x},t)}{\partial \bm{x}}\bm{f}(\bm{x},\bm{d})\leq 0, \forall \bm{d}\in \mathcal{D}$ for almost all  $(\bm{x},t)\in \overline{\mathcal{S}} \times [0,T]$. According to Lemma B.5 in \cite{lin1996smooth}, for arbitrary $\epsilon$ satisfying $0<\epsilon< \frac{\delta}{2(2T+3)}$, there exists a continuously differentiable function $w(\cdot,\cdot)\colon \mathbb{R}^n\times [0,T]\rightarrow \mathbb{R}$ such that 
\begin{equation*}
        \begin{cases}
            |w(\bm{x},t)-V(\bm{x},t)|\leq \epsilon, & \forall (\bm{x},t)\in \overline{\mathcal{S}}\times [0,T],\\
            \mathcal{L}^{\bm{d}} w(\bm{x},t)\leq \epsilon,& \forall (\bm{x},t) \in \overline{\mathcal{S}}\times [0,T], \forall \bm{d}\in \mathcal{D}.
        \end{cases}
    \end{equation*}
Let $w'(\bm{x},t)\colon=w(\bm{x},t)-2\epsilon t+2(T+1) \epsilon$. Obviously, $w'(\bm{x},t)-\epsilon+2\epsilon t-2T \epsilon=w(\bm{x},t)+\epsilon \geq V(\bm{x},t)$. Therefore, for $\bm{x}\in \mathcal{X}_0$, we have 
\[\begin{split}
  w'(\bm{x},0)&=w(\bm{x},0)+2(T+1)\epsilon \leq V(\bm{x},0)+\epsilon+2(T+1)\epsilon\\
  &\leq -\delta + \frac{\delta}{2(2T+3)}(2T+3)\leq -\frac{\delta}{2}.
\end{split}
\]
Since $V(\bm{x},t)\geq h(\bm{x})$ for $(\bm{x},t)\in \overline{\mathcal{S}}\times [0,T]$, we can obtain, for $(\bm{x},t)\in \overline{\mathcal{S}}\times [0,T]$,  
\[\begin{split}
    h(\bm{x})-w'(\bm{x},t)&=h(\bm{x})-w(\bm{x},t)+2\epsilon t-2(T+1) \epsilon\\&
    \leq h(\bm{x})-V(\bm{x},t)+\epsilon+2\epsilon t-2(T+1) \epsilon\\
    &\leq \epsilon-2\epsilon=-\epsilon.
\end{split}\]
Furthermore, we can obtain, for $(\bm{x},t) \in \overline{\mathcal{S}}\times [0,T]$, 
$\mathcal{L}^{\bm{d}} w'(\bm{x},t)=\mathcal{L}^{\bm{d}}w(\bm{x},t)-2\epsilon\leq \epsilon-2\epsilon=-\epsilon, \forall \bm{d}\in \mathcal{D}$.
Therefore, $w'(\cdot,\cdot)\colon \mathbb{R}^n\times [0,T] \rightarrow \mathbb{R}$ satisfies \eqref{sufficient}.
\end{proof}

If $v(\cdot,\cdot)\colon \mathbb{R}^n\times [0,T] \rightarrow \mathbb{R}$ satisfies $h(\bm{x})-v(\bm{x},t)\leq  0, \forall (\bm{x},t)\in \overline{\mathcal{S}} \times [0,T]$, it satisfies $v(\bm{x},t)\geq 0,  \forall (\bm{x},t)\in \partial \mathcal{S} \times [0,T]$. Thus, if there exists a continuously differentiable function $v(\cdot,\cdot): \mathbb{R}^n\times [0,T]\rightarrow \mathbb{R}$ satisfying \eqref{sufficient}, it is a time-dependent barrier certificate. Therefore, we obtain the most important result in this paper, which is formulated in Theorem \ref{converse}. 
\begin{theorem}
\label{converse}
The system \eqref{sys} is safe  over the time horizon $[0,T]$  if and only if there exists a time-dependent barrier certificate. 
\end{theorem}

In this study, we stipulate that the safe set $\mathcal{S}$ is bounded and open, and this assumption cannot be relaxed. Specifically, if the safe set $\mathcal{S}$ were compact, the existence of a continuously differentiable, time-dependent barrier certificate ensuring system safety over the time horizon $[0,T]$ could not, in general, be guaranteed. This limitation arises from fundamental mathematical constraints in our framework. When $\mathcal{S}$ is compact, if the system \eqref{sys} starting from $\bm{x}_0\in \mathcal{X}_0$ is safe over the time horizon $[0,T]$, we can only guarantee $V(\bm{x}_0,0)\leq 0$ rather than the strict inequality $V(\bm{x}_0,0)<0$ established in Lemma \ref{equiv} for open safe sets. Consequently, we cannot guarantee the existence of $\delta>0$ such that $V(\bm{x}_0,0)\leq -\delta$ for all $\bm{x}_0\in \mathcal{X}_0$ as required in Lemma \ref{equiv1}. This uniform negativity bound is essential for constructing the continuously differentiable approximation in Lemma \ref{converse_pre}, which forms the foundation of our converse barrier certificate theorem. Therefore, the assumption that $\mathcal{S}$ is open and bounded is fundamental to our theoretical framework and cannot be relaxed to compactness while maintaining the existence guarantees for continuously differentiable barrier certificates.

\color{black}

\begin{remark} 
We observe that the system \eqref{sys} is safe over $[0,T]$ if and only if a continuously differentiable function $v(\cdot,\cdot)\colon \mathbb{R}^n \times [0,T]\rightarrow \mathbb{R}$ satisfies 
\begin{equation}
\label{Finite_safe1}
    \begin{cases}
v(\bm{x},0)< 0, & \forall \bm{x}\in \mathcal{X}_0,\\
v(\bm{x},t)\geq 0,     & \forall (\bm{x},t)\in \partial \mathcal{S}\times [0,T],\\
\mathcal{L}^{\bm{d}}v(\bm{x},t)\leq \lambda v(\bm{x},t), & \forall (\bm{x},t) \in \overline{\mathcal{S}}\times [0,T], \forall \bm{d}\in \mathcal{D},
    \end{cases}
    \end{equation}
where $\lambda\in \mathbb{R}$ is a user-defined value. If a function $v(\cdot,\cdot)\colon \mathbb{R}^n \times [0,T]\rightarrow \mathbb{R}$ satisfying \eqref{Finite_safe1} exists, the system \eqref{sys} is safe over the time horizon $[0,T]$. On the other hand, if the system \eqref{sys} is safe over the time horizon $[0,T]$,  there exists a continuously differentiable function $v$ satisfying \eqref{Finite_safe}. Via taking $v'(\bm{x},t): =e^{\lambda t} v(\bm{x},t)$, we can obtain $v'$ satisfies \eqref{Finite_safe1}. In addition, it is worth noting that when $\lambda < 0$, the function $v$ satisfying \eqref{Finite_safe1} is not required to be monotonic over the set $\{(\bm{x},t)\in \mathcal{S}\times [0,T]\mid v(\bm{x},t)<0\}$. However, it should be strictly monotonically decreasing along the system dynamics over $\{(\bm{x},t)\in \mathcal{S}\times [0,T]\mid v(\bm{x},t)>0\}$ . Conversely, when $\lambda > 0$, the requirement is exactly opposite. In the case where $\lambda = 0$, \eqref{Finite_safe1} is equivalent to \eqref{Finite_safe}, and $v$ satisfying \eqref{Finite_safe1} should be monotonically non-increasing over $\mathbb{R}^n\times [0,T]$. Thus,  the constraints corresponding to these three scenarios may exhibit different performances in safety verification practices.  
\end{remark}

\color{black}
\begin{remark}
The results presented in this work can be extended to alternative constructions that employ marginal functions as introduced in \cite{maghenem2022converse}. Please refer to Proposition \ref{thm:marginal-function} in Appendix B.
\end{remark}

\color{black}
\begin{remark}[Connection to Infinite-Time Safety via Time Augmentation]
\label{remark:time-augmentation}
The finite-time safety problem considered in this work admits a systematic reformulation as an infinite-time robust safety problem through state augmentation. Consider the extended system:
\[
\dot{\boldsymbol{z}} = \boldsymbol{g}(\boldsymbol{z},\boldsymbol{d}) = \begin{bmatrix} \boldsymbol{f}(\boldsymbol{x},\boldsymbol{d}) \\ 1 \end{bmatrix}, \quad \boldsymbol{z} = \begin{bmatrix} \boldsymbol{x} \\ \tau \end{bmatrix} \in \mathbb{R}^{n+1}
\]
with the extended safe set $\mathcal{S}_{\text{ext}} = \{(\boldsymbol{x},\tau) \in \mathbb{R}^n \times \mathbb{R} \mid \tau \in [0,T], \boldsymbol{x} \in \mathbb{R}^n\}$ and initial set $\mathcal{X}_{0,\text{ext}} = \{(\boldsymbol{x},0) \mid \boldsymbol{x} \in \mathcal{X}_0\}$. In this formulation, finite-time safety over $[0,T]$ corresponds precisely to infinite-time safety for the augmented system.

This perspective reveals important connections to established converse theorems for robust safety. Compared to \cite{maghenem2022converse}, which established the existence of lower semicontinuous and locally Lipschitz barrier functions for general differential inclusions, our framework yields \emph{continuously differentiable} barrier certificates. This enhanced regularity stems from our focus on finite-time horizons and the specific structure of perturbed deterministic systems, enabling stronger regularity results through Hamilton-Jacobi approximation techniques.  Furthermore, while \cite{10316551} introduces a strong robust-safety notion for differential inclusions with continuously differentiable barrier certificates, their conditions rely on set-valued analysis and Clarke generalized gradients that pose significant computational challenges. In contrast, our approach employs conditions that are directly amenable to established numerical methods, such as sum-of-squares programming for polynomial systems, as demonstrated in our computational example.

Our work thus complements existing converse theorems by addressing finite-time safety with enhanced regularity guarantees and improved computational tractability, while maintaining a formulation that aligns directly with practical finite-horizon safety verification problems.
\end{remark}

\section{Examples}
In this section, we demonstrate the application of our theoretical developments through one example. In this example, the function $\bm{f}$ is polynomial in $\bm{x}$ and affine in the disturbance, and the safe set is semi-algebraic, $\mathcal{S}=\{\bm{x}\in \mathbb{R}^n\mid h(\bm{x})<0\}$ with $h(\bm{x})$ a polynomial. We therefore search for polynomial barrier certificates that satisfy the sufficient conditions in \eqref{sufficient}. To do this, we encode the constraint \eqref{sufficient} as a semi-definite program (SDP) using the sum of squares (SOS) decomposition for multivariate polynomials. The resulting SDP is then solved using the tool Mosek 10.1.21 \cite{aps2019mosek}. To ensure numerical stability during the solution of this SDP, we impose a constraint on the coefficients of the unknown polynomials, specifically restricting them to the interval $[-100, 100]$. In the sequel, $\sum[\bm{y}]$ denotes the set of sum-of-squares polynomials over variables $\bm{y}$, i.e., 
\[\sum[\bm{y}]=\{p\in \mathbb{R}[\bm{y}]\mid p=\sum_{i=1}^k q^2_i(\bm{y}), q_i(\bm{y})\in \mathbb{R}[\bm{y}],i=1,\ldots,k\},\] where $\mathbb{R}[\bm{y}]$ denotes the ring of polynomials in variables $\bm{y}$.

\begin{example}
\label{ex1}
Consider a nonlinear system with disturbance inputs, adapted from \cite{xue2023reach}:
\begin{equation*}
    \begin{cases}
        \dot{x} = -0.42x - 1.05y - 2.3x^2 - 0.56xy - x^3 \\
        \dot{y} = 1.98x + xy + d
    \end{cases}
\end{equation*}
with the following specifications:
\begin{itemize}
    \item Safe set: $\mathcal{S} = \{\bm{x} \in \mathbb{R}^2 \mid h(\bm{x}) < 0\}$, where $h(\bm{x}) = x^2 + y^2 - 4$;
    \item Initial set: $\mathcal{X}_0 = \{\bm{x} \in \mathbb{R}^2 \mid h_0(\bm{x}) \leq 0\}$, where $h_0(\bm{x}) = (x - 1.2)^2 + (y - 0.8)^2 - 0.1$;  
    \item Disturbance set: $\mathcal{D} = \{d \in \mathbb{R} \mid g(d) \leq 0\}$, where $g(d) = d^2 - 10^{-4}$;
    \item Time horizon: $T = 1.5$.
\end{itemize}

To solve the finite-time safety verification problem, we formulate the SDP derived from the barrier certificate conditions in \eqref{sufficient}:
\begin{equation*}
\begin{cases}
-v(\bm{x},0) - \epsilon + s_0(\bm{x}) h_0(\bm{x}) \in \sum[\bm{x}], \\
v(\bm{x},t) - p(\bm{x},t)h(\bm{x}) + s_1(\bm{x},t)t(t-T) \in \sum[\bm{x},t], \\
-\mathcal{L}^{\bm{d}}v(\bm{x},t) + s_2(\bm{x},\bm{d},t)h(\bm{x}) + s_3(\bm{x},\bm{d},t)g(\bm{d}) + s_4(\bm{x},\bm{d},t)t(t-T) \in \sum[\bm{x},\bm{d},t], \\
v(\bm{x},t), p(\bm{x},t) \in \mathbb{R}[\bm{x},t], \quad s_0(\bm{x}) \in \sum[\bm{x}], \quad s_1(\bm{x},t) \in \sum[\bm{x},t], \\
s_i(\bm{x},\bm{d},t) \in \sum[\bm{x},\bm{d},t], \quad i=2,3,4,
\end{cases}
\end{equation*}
where $\epsilon > 0$ is a numerical tolerance parameter and $\mathcal{L}^{\bm{d}}v(\bm{x},t):= \frac{\partial v(\bm{x},t)}{\partial t} + \frac{\partial v(\bm{x},t)}{\partial \bm{x}}\bm{f}(\bm{x},\bm{d})$.

The SDP is feasible with polynomial degrees set to 4 and $\epsilon = 10^{-6}$, yielding a valid barrier certificate $v(\bm{x},t)$. The computation time is 0.91 seconds on a workstation with an Intel Core i7-7500U 2.70 GHz CPU and 32 GB of RAM running Windows 10, demonstrating the practical efficiency of our approach. This numerically certifies that the system is safe over the time horizon $[0,T]$.

\begin{figure}[htbp]
    \centering
    \includegraphics[width=0.5\textwidth]{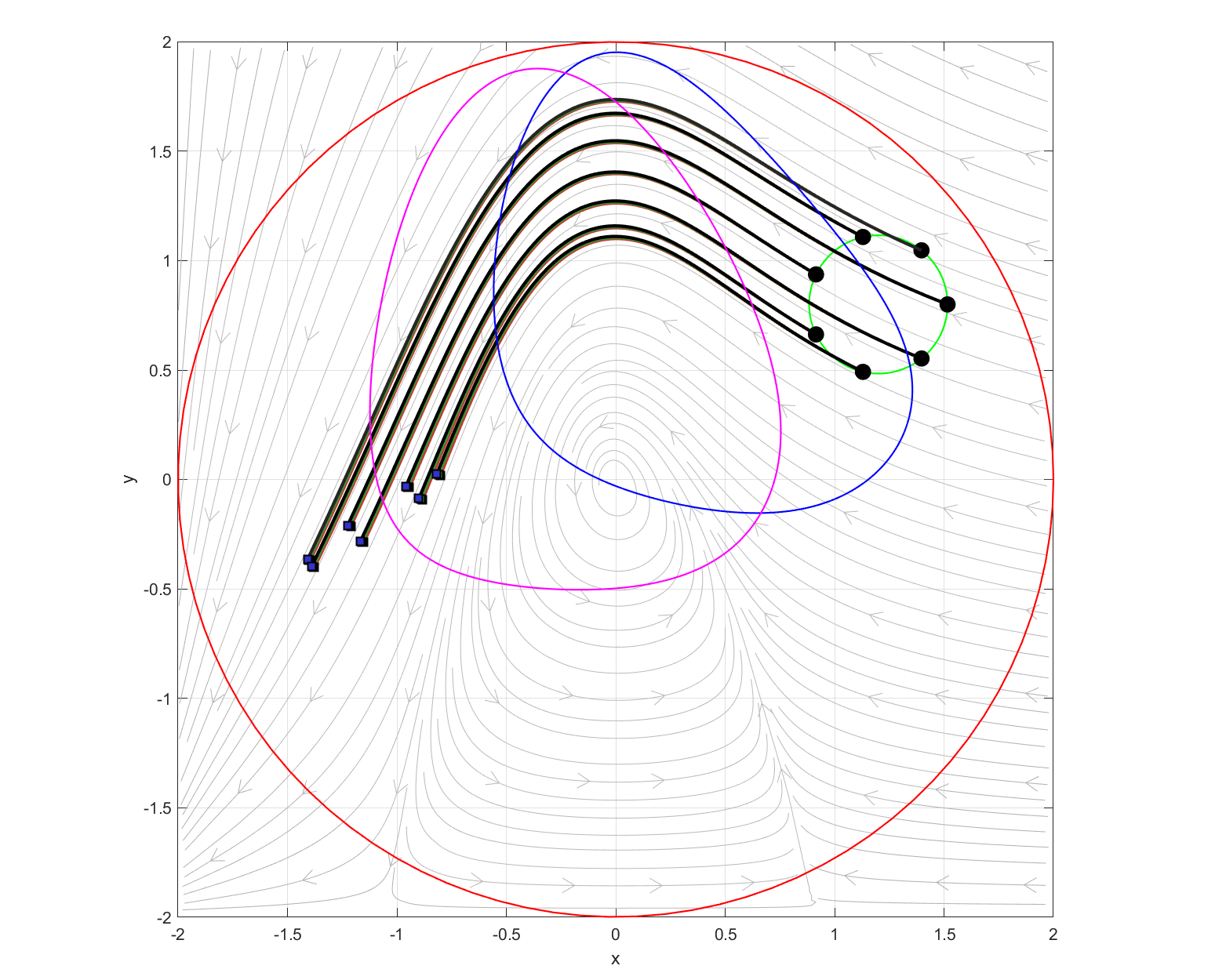} 
    \caption{Safety verification via barrier certificates for Example~\ref{ex1}. The green curve denotes the initial set boundary $\partial\mathcal{X}_0$, the red curve shows the safe set boundary $\partial\mathcal{S}$, while the blue and pink curves represent the zero level sets of the barrier certificate at $t=0.5$ and $t=1.0$, respectively. All trajectories originating from $\mathcal{X}_0$ remain within the certified safe region bounded by the barrier level sets.}
    \label{fig:ex1_traj}
\end{figure}
\end{example}

\color{black}

\section{Conclusion}
\label{con}
This paper establishes that the existence of a continuously differentiable, time-dependent barrier certificate is both a necessary and sufficient condition for guaranteeing the finite-time safety of a class of continuous-time deterministic systems subject to measurable disturbances, thereby providing a fundamental completeness guarantee for finite-time safety verification. The proof is developed within the framework of Hamilton–Jacobi reachability.

 Future work will focus on developing advanced numerical algorithms to efficiently solve the proposed barrier conditions.

\color{black}

\section*{Acknowledgements}
This work is funded by the CAS Pioneer Hundred Talents Program, Basic Research Program of  Institute of Software, CAS (Grant No. ISCAS-JCMS-202302), the NSFC under grant No. 12201037, and NRF RSS Scheme NRF-RSS2022-009.

\section*{Appendix A}
\color{red}

\color{black}

\textbf{The proof of Proposition \ref{prop1}:}
\begin{proof}

 Let $v$ be a time-dependent barrier certificate satisfying \eqref{Finite_safe}. Consider any initial state $\bm{x}_0 \in \mathcal{X}_0$, disturbance signal $\bm{d}(\cdot) \in \mathcal{M}$, and corresponding trajectory $\bm{\phi}_{\bm{x}_0}^{\bm{d}}(t)$. Define $\gamma(t) = v(\bm{\phi}_{\bm{x}_0}^{\bm{d}}(t), t)$.

Since $v$ is continuously differentiable and the trajectory is absolutely continuous, $\gamma(t)$ is absolutely continuous with derivative:
\[
\frac{d\gamma(t)}{dt} = \mathcal{L}^{\bm{d}}v(\bm{\phi}_{\bm{x}_0}^{\bm{d}}(t), t) \leq 0 \quad \text{for almost all } t \in [0,T].
\]
Thus, $\gamma(t)$ is non-increasing. From condition \eqref{Finite_safe}, $\gamma(0) = v(\bm{x}_0, 0) < 0$, so $\gamma(t) < 0$ for all $t \in [0,T]$.

Suppose for contradiction that there exists $t^* \in [0,T]$ with $\bm{\phi}_{\bm{x}_0}^{\bm{d}}(t^*) \in \partial\mathcal{S}$. Then $\gamma(t^*) = v(\bm{\phi}_{\bm{x}_0}^{\bm{d}}(t^*), t^*) \geq 0$, contradicting $\gamma(t^*) < 0$. Hence, the trajectory remains within $\mathcal{S}$ for all $t \in [0,T]$. 
\end{proof}

\textbf{The proof of Lemma \ref{continuous}:}
\begin{proof}[Proof of Lemma \ref{continuous}]
Let \((\boldsymbol{x}_0, t_0) \in \mathbb{R}^n \times [0,T]\) be an arbitrary point. We will show that there exists a neighborhood of \((\boldsymbol{x}_0, t_0)\) on which \(V\) is Lipschitz continuous.

Since the vector field \(\boldsymbol{F}\) is globally Lipschitz, it satisfies a linear growth condition: there exists \(C>0\) such that \(\|\boldsymbol{F}(\boldsymbol{x},\boldsymbol{d})\| \leq C(1+\|\boldsymbol{x}\|)\) for all \(\boldsymbol{x} \in \mathbb{R}^n\) and \(\boldsymbol{d} \in \mathcal{D}\). This implies that the trajectories of the system satisfy
\[
\|\boldsymbol{\psi}_{\boldsymbol{x}}^{\boldsymbol{d}}(s)\| \leq (\|\boldsymbol{x}\| + CT)e^{CT}, \quad \forall s \in [0,T].
\]
Let \(R = (\|\boldsymbol{x}_0\| + CT)e^{CT} + 1\). Then for any \(\boldsymbol{x}\) in the closed ball \(B(\boldsymbol{x}_0, 1)\) and \(t \in [0,T]\), the trajectory \(\boldsymbol{\psi}_{\boldsymbol{x}}^{\boldsymbol{d}}(s)\) for \(s \in [0,T]\) remains in the closed ball \(B(0, R)\). Let \(\mathcal{X} = \overline{B(0, R)}\), which is compact.

Since \(h\) is locally Lipschitz, it is Lipschitz continuous on \(\mathcal{X}\) with some constant \(L_h\). Also, by the global Lipschitz property of \(\boldsymbol{F}\), there exists \(L_F>0\) such that for all \(\boldsymbol{x}_1, \boldsymbol{x}_2 \in \mathbb{R}^n\), \(\boldsymbol{d} \in \mathcal{D}\), and \(s \in [0,T]\):
\[
\| \boldsymbol{\psi}^{\boldsymbol{d}}_{\boldsymbol{x}_1}(s) - \boldsymbol{\psi}^{\boldsymbol{d}}_{\boldsymbol{x}_2}(s) \| \leq e^{L_F s} \| \boldsymbol{x}_1 - \boldsymbol{x}_2 \| \leq e^{L_F T} \| \boldsymbol{x}_1 - \boldsymbol{x}_2 \|.
\]

Furthermore, by compactness of \(\mathcal{X} \times \mathcal{D}\), there exists \(M_F > 0\) such that \(\| \boldsymbol{F}(\boldsymbol{x}, \boldsymbol{d}) \| \leq M_F\) for all \((\boldsymbol{x}, \boldsymbol{d}) \in \mathcal{X} \times \mathcal{D}\), which implies that the flow is Lipschitz in time:
\[
\| \boldsymbol{\psi}^{\boldsymbol{d}}_{\boldsymbol{x}}(s_1) - \boldsymbol{\psi}^{\boldsymbol{d}}_{\boldsymbol{x}}(s_2) \| \leq M_F |s_1 - s_2|, \quad \forall s_1, s_2 \in [0,T].
\]

Now, consider any \((\boldsymbol{x}_1, t_1), (\boldsymbol{x}_2, t_2) \in B(\boldsymbol{x}_0, 1) \times [0,T]\). We analyze the Lipschitz continuity in both variables:

\textbf{Lipschitz Continuity in \(\boldsymbol{x}\):}

Fix \(t \in [0,T]\) and consider \(\boldsymbol{x}_1, \boldsymbol{x}_2 \in B(\boldsymbol{x}_0, 1)\). For any \(\epsilon > 0\), there exists \(\boldsymbol{d}_1 \in \mathcal{M}\) and \(\tau_1 \in [t,T]\) such that
\[
V(\boldsymbol{x}_1, t) \leq h(\boldsymbol{\psi}^{\boldsymbol{d}_1}_{\boldsymbol{x}_1}(\tau_1 - t)) + \epsilon.
\]

Then, 
\[
V(\boldsymbol{x}_1, t) - V(\boldsymbol{x}_2, t) \leq h(\boldsymbol{\psi}^{\boldsymbol{d}_1}_{\boldsymbol{x}_1}(\tau_1 - t)) - h(\boldsymbol{\psi}^{\boldsymbol{d}_1}_{\boldsymbol{x}_2}(\tau_1 - t)) + \epsilon.
\]

Since \(h\) is Lipschitz on \(\mathcal{X}\) with constant \(L_h\), we have
\[
|h(\boldsymbol{\psi}^{\boldsymbol{d}_1}_{\boldsymbol{x}_1}(\tau_1 - t)) - h(\boldsymbol{\psi}^{\boldsymbol{d}_1}_{\boldsymbol{x}_2}(\tau_1 - t))| \leq L_h \| \boldsymbol{\psi}^{\boldsymbol{d}_1}_{\boldsymbol{x}_1}(\tau_1 - t) - \boldsymbol{\psi}^{\boldsymbol{d}_1}_{\boldsymbol{x}_2}(\tau_1 - t) \|.
\]

Using the flow bound
\[
\| \boldsymbol{\psi}^{\boldsymbol{d}_1}_{\boldsymbol{x}_1}(\tau_1 - t) - \boldsymbol{\psi}^{\boldsymbol{d}_1}_{\boldsymbol{x}_2}(\tau_1 - t) \| \leq e^{L_F T} \| \boldsymbol{x}_1 - \boldsymbol{x}_2 \|,
\]
we obtain
\[
V(\boldsymbol{x}_1, t) - V(\boldsymbol{x}_2, t) \leq L_h e^{L_F T} \| \boldsymbol{x}_1 - \boldsymbol{x}_2 \| + \epsilon.
\]

By symmetry and the arbitrariness of \(\epsilon\), we have
\[
|V(\boldsymbol{x}_1, t) - V(\boldsymbol{x}_2, t)| \leq L_h e^{L_F T} \| \boldsymbol{x}_1 - \boldsymbol{x}_2 \|.
\]

\textbf{Lipschitz Continuity in \(t\):}

Fix \(\boldsymbol{x} \in B(\boldsymbol{x}_0, 1)\) and consider \(t_1, t_2 \in [0,T]\) with \(t_1 \leq t_2\). For any \(\epsilon > 0\), there exist \(\boldsymbol{d}_1 \in \mathcal{M}\) and \(\tau_1 \in [t_1, T]\) such that 
\[
V(\boldsymbol{x}, t_1) \leq h(\boldsymbol{\psi}^{\boldsymbol{d}_1}_{\boldsymbol{x}}(\tau_1 - t_1)) + \epsilon.
\]

We consider two cases:

\textbf{Case 1:} \(\tau_1 \in [t_2, T]\). Then:
\[
V(\boldsymbol{x}, t_2) \geq h(\boldsymbol{\psi}^{\boldsymbol{d}_1}_{\boldsymbol{x}}(\tau_1 - t_2)).
\]
Thus, we have
\[
V(\boldsymbol{x}, t_1) - V(\boldsymbol{x}, t_2) \leq h(\boldsymbol{\psi}^{\boldsymbol{d}_1}_{\boldsymbol{x}}(\tau_1 - t_1)) - h(\boldsymbol{\psi}^{\boldsymbol{d}_1}_{\boldsymbol{x}}(\tau_1 - t_2)) + \epsilon.
\]

\textbf{Case 2:} \(\tau_1 \in [t_1, t_2]\). Let \(\tau_2 = t_2\). Then, we have
\[
V(\boldsymbol{x}, t_2) \geq h(\boldsymbol{\psi}^{\boldsymbol{d}_1}_{\boldsymbol{x}}(\tau_2 - t_2)) = h(\boldsymbol{\psi}^{\boldsymbol{d}_1}_{\boldsymbol{x}}(0)) = h(\boldsymbol{x}).
\]
Thus, we have
\[
V(\boldsymbol{x}, t_1) - V(\boldsymbol{x}, t_2) \leq h(\boldsymbol{\psi}^{\boldsymbol{d}_1}_{\boldsymbol{x}}(\tau_1 - t_1)) - h(\boldsymbol{x}) + \epsilon.
\]

In both cases, using the Lipschitz property of \(h\) and the time-Lipschitz bound of the flow
\[
|h(\boldsymbol{\psi}^{\boldsymbol{d}_1}_{\boldsymbol{x}}(s_1)) - h(\boldsymbol{\psi}^{\boldsymbol{d}_1}_{\boldsymbol{x}}(s_2))| \leq L_h M_F |s_1 - s_2|,
\]
we obtain
\[
V(\boldsymbol{x}, t_1) - V(\boldsymbol{x}, t_2) \leq L_h M_F |t_2 - t_1| + \epsilon.
\]

Since \(\epsilon\) is arbitrary, we have
\[
|V(\boldsymbol{x}, t_1) - V(\boldsymbol{x}, t_2)| \leq L_h M_F |t_1 - t_2|.
\]

Combining both results, we obtain, for all \((\boldsymbol{x}_1, t_1), (\boldsymbol{x}_2, t_2) \in B(\boldsymbol{x}_0, 1) \times [0,T]\), 
\[
|V(\boldsymbol{x}_1, t_1) - V(\boldsymbol{x}_2, t_2)| \leq L \left( \| \boldsymbol{x}_1 - \boldsymbol{x}_2 \| + |t_1 - t_2| \right),
\]
where \(L = \max\{L_h e^{L_F T}, L_h M_F\}\). This shows that \(V\) is Lipschitz continuous on \(B(\boldsymbol{x}_0, 1) \times [0,T]\), and hence locally Lipschitz on \(\mathbb{R}^n \times [0,T]\).

This completes the proof of Lemma \ref{continuous}. 
\end{proof}

\color{black}

\textbf{The proof of Proposition \ref{HJ}:}

\begin{proof}
We establish that $V$ is the unique Lipschitz viscosity solution to the Hamilton-Jacobi equation (\ref{hami}).

\textbf{Step 1: Viscosity Solution Property}

Let $\phi \in C^{1}(\mathbb{R}^{n} \times [0,T])$ be a test function, i.e., $\phi$ is continuously differentiable on $\mathbb{R}^{n} \times [0,T]$.

\textbf{Supersolution:} Suppose $V - \phi$ attains a local minimum at $(\boldsymbol{x}_0, t_0)$. Without loss of generality, assume $V(\boldsymbol{x}_0, t_0) = \phi(\boldsymbol{x}_0, t_0)$. Then for sufficiently small $\delta > 0$ and any $\boldsymbol{d} \in \mathcal{M}$, we have
\[
V(\boldsymbol{x}_0, t_0) - \phi(\boldsymbol{x}_0, t_0) \leq V(\boldsymbol{\psi}_{\boldsymbol{x}_0}^{\boldsymbol{d}}(\delta), t_0 + \delta) - \phi(\boldsymbol{\psi}_{\boldsymbol{x}_0}^{\boldsymbol{d}}(\delta), t_0 + \delta).
\]
Rearranging and using $V(\boldsymbol{x}_0, t_0) = \phi(\boldsymbol{x}_0, t_0)$ gives
\[
\phi(\boldsymbol{x}_0, t_0) - \phi(\boldsymbol{\psi}_{\boldsymbol{x}_0}^{\boldsymbol{d}}(\delta), t_0 + \delta) \geq V(\boldsymbol{x}_0, t_0) - V(\boldsymbol{\psi}_{\boldsymbol{x}_0}^{\boldsymbol{d}}(\delta), t_0 + \delta).
\]
From the dynamic programming principle, $V(\boldsymbol{x}_0, t_0) \geq V(\boldsymbol{\psi}_{\boldsymbol{x}_0}^{\boldsymbol{d}}(\delta), t_0 + \delta)$, so the right-hand side is non-negative. Thus, we obtain
\[
\phi(\boldsymbol{x}_0, t_0) - \phi(\boldsymbol{\psi}_{\boldsymbol{x}_0}^{\boldsymbol{d}}(\delta), t_0 + \delta) \geq 0.
\]
Dividing by $\delta > 0$ and taking $\delta \to 0^+$ yields
\[
\frac{\partial \phi(\boldsymbol{x}_0,t_0)}{\partial t} + \frac{\partial \phi(\boldsymbol{x}_0,t_0)}{\partial \boldsymbol{x}}\boldsymbol{F}(\boldsymbol{x}_0,\boldsymbol{d}) \leq 0,
\]
which implies
\[
\frac{\partial \phi(\boldsymbol{x}_0,t_0)}{\partial t} + \frac{\partial \phi(\boldsymbol{x}_0,t_0)}{\partial \boldsymbol{x}}\boldsymbol{F}(\boldsymbol{x}_0,\boldsymbol{d}) \leq 0, \quad \forall \boldsymbol{d} \in \mathcal{D}.
\]
Taking supremum over $\boldsymbol{d}$ gives
\[
\frac{\partial \phi(\boldsymbol{x}_0,t_0)}{\partial t} + \sup_{\boldsymbol{d}\in\mathcal{D}}\frac{\partial \phi(\boldsymbol{x}_0,t_0)}{\partial \boldsymbol{x}}\boldsymbol{F}(\boldsymbol{x}_0,\boldsymbol{d}) \leq 0.
\]
Additionally, since $V(\boldsymbol{x},t) \geq h(\boldsymbol{x})$ by definition, we have $h(\boldsymbol{x}_0) - V(\boldsymbol{x}_0,t_0) \leq 0$. Therefore, the supersolution condition holds.

\textbf{Subsolution:} Suppose $V - \phi$ attains a local maximum at $(\boldsymbol{x}_0, t_0)$, and assume $V(\boldsymbol{x}_0, t_0) = \phi(\boldsymbol{x}_0, t_0)$. We want to show that:
\[
\max\left\{ h(\boldsymbol{x}_0) - V(\boldsymbol{x}_0,t_0), \frac{\partial \phi(\boldsymbol{x}_0,t_0)}{\partial t} + \sup_{\boldsymbol{d}\in\mathcal{D}}\frac{\partial \phi(\boldsymbol{x}_0,t_0)}{\partial \boldsymbol{x}}\boldsymbol{F}(\boldsymbol{x}_0,\boldsymbol{d}) \right\} \geq 0.
\]
Assume, by contradiction, that this is false. Then there exists $\epsilon_1, \epsilon_2 > 0$ such that
\[
h(\boldsymbol{x}_0) - V(\boldsymbol{x}_0,t_0) \leq -\epsilon_1
\]
and
\[
\frac{\partial \phi(\boldsymbol{x}_0,t_0)}{\partial t} + \sup_{\boldsymbol{d}\in\mathcal{D}}\frac{\partial \phi(\boldsymbol{x}_0,t_0)}{\partial \boldsymbol{x}}\boldsymbol{F}(\boldsymbol{x}_0,\boldsymbol{d}) \leq -\epsilon_2.
\]
By continuity of $h$, $V$, $\phi$, and $\boldsymbol{F}$, there exists $\delta'>0$ such that for all $(\boldsymbol{x},t)$ with $\|\boldsymbol{x}-\boldsymbol{x}_0\| \leq \delta'$ and $|t-t_0| \leq \delta'$, we have
\[
h(\boldsymbol{x}) - V(\boldsymbol{x},t) \leq -\epsilon_1/2
\]
and
\[
\frac{\partial \phi(\boldsymbol{x},t)}{\partial t} + \sup_{\boldsymbol{d}\in\mathcal{D}}\frac{\partial \phi(\boldsymbol{x},t)}{\partial \boldsymbol{x}}\boldsymbol{F}(\boldsymbol{x},\boldsymbol{d}) \leq -\epsilon_2/2.
\]
In particular, for any $\boldsymbol{d} \in \mathcal{D}$ and $(\boldsymbol{x},t)$ in this neighborhood,
\[
\frac{\partial \phi(\boldsymbol{x},t)}{\partial t} + \frac{\partial \phi(\boldsymbol{x},t)}{\partial \boldsymbol{x}}\boldsymbol{F}(\boldsymbol{x},\boldsymbol{d}) \leq -\epsilon_2/2.
\]

Now, by the dynamic programming principle, for any $\delta \in (0, \delta']$ such that the trajectory $\boldsymbol{\psi}_{\boldsymbol{x}_0}^{\boldsymbol{d}}(s)$ remains in the neighborhood for $s \in [t_0, t_0+\delta]$, we have
\[
V(\boldsymbol{x}_0, t_0) = \sup_{\boldsymbol{d} \in \mathcal{M}_{[t_0,t_0+\delta]}} \max\left\{ V(\boldsymbol{\psi}_{\boldsymbol{x}_0}^{\boldsymbol{d}}(t_0+\delta), t_0+\delta), \sup_{s \in [t_0,t_0+\delta]} h(\boldsymbol{\psi}_{\boldsymbol{x}_0}^{\boldsymbol{d}}(s)) \right\}.
\]
Then, for any $\epsilon > 0$, there exists $\boldsymbol{d}^* \in \mathcal{M}_{[t_0,t_0+\delta]}$ such that
\[
V(\boldsymbol{x}_0, t_0) \leq \max\left\{ V(\boldsymbol{\psi}_{\boldsymbol{x}_0}^{\boldsymbol{d}^*}(t_0+\delta), t_0+\delta), \sup_{s \in [t_0,t_0+\delta]} h(\boldsymbol{\psi}_{\boldsymbol{x}_0}^{\boldsymbol{d}^*}(s)) \right\} + \epsilon\delta.
\]

We now consider two cases:

\textbf{Case 1:} 
\[
\sup_{s \in [t_0,t_0+\delta]} h(\boldsymbol{\psi}_{\boldsymbol{x}_0}^{\boldsymbol{d}^*}(s)) \geq V(\boldsymbol{x}_0, t_0) - \epsilon\delta.
\]
Then, there exists $s_0 \in [t_0, t_0+\delta]$ such that
\[
h(\boldsymbol{\psi}_{\boldsymbol{x}_0}^{\boldsymbol{d}^*}(s_0)) \geq V(\boldsymbol{x}_0, t_0) - \epsilon\delta.
\]
But by our assumption, in the neighborhood we have $h(\boldsymbol{\psi}_{\boldsymbol{x}_0}^{\boldsymbol{d}^*}(s_0)) \leq V(\boldsymbol{\psi}_{\boldsymbol{x}_0}^{\boldsymbol{d}^*}(s_0), s_0) - \epsilon_1/2$.
Also, since $V - \phi$ has a local maximum at $(\boldsymbol{x}_0, t_0)$, we have $V(\boldsymbol{\psi}_{\boldsymbol{x}_0}^{\boldsymbol{d}^*}(s_0), s_0) \leq \phi(\boldsymbol{\psi}_{\boldsymbol{x}_0}^{\boldsymbol{d}^*}(s_0), s_0)$.
Now, consider the function $s \mapsto \phi(\boldsymbol{\psi}_{\boldsymbol{x}_0}^{\boldsymbol{d}^*}(s), s)$. Its derivative satisfies
\[
\frac{d}{ds} \phi(\boldsymbol{\psi}_{\boldsymbol{x}_0}^{\boldsymbol{d}^*}(s), s) = \frac{\partial \phi}{\partial t} + \frac{\partial \phi}{\partial \boldsymbol{x}} \boldsymbol{F}(\boldsymbol{\psi}_{\boldsymbol{x}_0}^{\boldsymbol{d}^*}(s), \boldsymbol{d}^*(s)) \leq -\epsilon_2/2.
\]
Integrating from $t_0$ to $s_0$, we get
\[
\phi(\boldsymbol{\psi}_{\boldsymbol{x}_0}^{\boldsymbol{d}^*}(s_0), s_0) - \phi(\boldsymbol{x}_0, t_0) \leq -\frac{\epsilon_2}{2} (s_0 - t_0).
\]
Since $\phi(\boldsymbol{x}_0, t_0) = V(\boldsymbol{x}_0, t_0)$, we have
\[
V(\boldsymbol{\psi}_{\boldsymbol{x}_0}^{\boldsymbol{d}^*}(s_0), s_0) \leq \phi(\boldsymbol{\psi}_{\boldsymbol{x}_0}^{\boldsymbol{d}^*}(s_0), s_0) \leq V(\boldsymbol{x}_0, t_0) - \frac{\epsilon_2}{2} (s_0 - t_0).
\]
Therefore,
\[
h(\boldsymbol{\psi}_{\boldsymbol{x}_0}^{\boldsymbol{d}^*}(s_0)) \leq V(\boldsymbol{x}_0, t_0) - \frac{\epsilon_2}{2} (s_0 - t_0) - \epsilon_1/2.
\]
But we also have $h(\boldsymbol{\psi}_{\boldsymbol{x}_0}^{\boldsymbol{d}^*}(s_0)) \geq V(\boldsymbol{x}_0, t_0) - \epsilon\delta$.
Combining, we get
\[
V(\boldsymbol{x}_0, t_0) - \epsilon\delta \leq V(\boldsymbol{x}_0, t_0) - \frac{\epsilon_2}{2} (s_0 - t_0) - \epsilon_1/2,
\]
and thus, 
\[
\frac{\epsilon_2}{2} (s_0 - t_0) + \epsilon_1/2 \leq \epsilon\delta.
\]
Since $s_0 - t_0 \geq 0$, we have $\epsilon_1/2 \leq \epsilon\delta$. If we choose $\epsilon < \epsilon_1/(2\delta)$, then we get a contradiction.

\textbf{Case 2:}
\[
\sup_{s \in [t_0,t_0+\delta]} h(\boldsymbol{\psi}_{\boldsymbol{x}_0}^{\boldsymbol{d}^*}(s)) < V(\boldsymbol{x}_0, t_0) - \epsilon\delta.
\]
Then, from the dynamic programming principle, we have
\[
V(\boldsymbol{x}_0, t_0) \leq V(\boldsymbol{\psi}_{\boldsymbol{x}_0}^{\boldsymbol{d}^*}(t_0+\delta), t_0+\delta) + \epsilon\delta.
\]
Again, using the test function, we have
\[
V(\boldsymbol{\psi}_{\boldsymbol{x}_0}^{\boldsymbol{d}^*}(t_0+\delta), t_0+\delta) \leq \phi(\boldsymbol{\psi}_{\boldsymbol{x}_0}^{\boldsymbol{d}^*}(t_0+\delta), t_0+\delta).
\]
And as before,
\[
\phi(\boldsymbol{\psi}_{\boldsymbol{x}_0}^{\boldsymbol{d}^*}(t_0+\delta), t_0+\delta) - \phi(\boldsymbol{x}_0, t_0) \leq -\frac{\epsilon_2}{2} \delta.
\]
So,
\[
V(\boldsymbol{\psi}_{\boldsymbol{x}_0}^{\boldsymbol{d}^*}(t_0+\delta), t_0+\delta) \leq V(\boldsymbol{x}_0, t_0) - \frac{\epsilon_2}{2} \delta.
\]
Then,
\[
V(\boldsymbol{x}_0, t_0) \leq V(\boldsymbol{x}_0, t_0) - \frac{\epsilon_2}{2} \delta + \epsilon\delta.
\]
If we choose $\epsilon < \epsilon_2/2$, then we get $V(\boldsymbol{x}_0, t_0) < V(\boldsymbol{x}_0, t_0)$, a contradiction.

Therefore, in both cases we obtain a contradiction. Hence, the subsolution condition must hold.

\textbf{Step 2: Uniqueness}

Uniqueness follows from the comparison principle for viscosity solutions.  Define the Hamiltonian $H(\cdot,\cdot,\cdot): \mathbb{R}^n\times \mathbb{R}\times \mathbb{R}^b \rightarrow \mathbb{R}$ by $H(\boldsymbol{x}, t, \boldsymbol{p})= \sup_{\boldsymbol{d} \in \mathcal{D}} \boldsymbol{p} \cdot \boldsymbol{F}(\boldsymbol{x}, \boldsymbol{d})$. It is 

1) Continuous in all arguments (since $\boldsymbol{F}$ is continuous and $\mathcal{D}$ is compact);

2) Convex in $\boldsymbol{p}$ (as a supremum of linear functions);

3) The terminal condition $h$ is locally Lipschitz.

Under these conditions, standard viscosity solution theory [16, Theorem 2.1] guarantees that equation (\ref{hami}) admits at most one Lipschitz continuous viscosity solution.

Since $V$ is Lipschitz continuous (by Lemma \ref{continuous}) and satisfies the viscosity solution conditions, it is the unique Lipschitz viscosity solution to (\ref{hami}).
\end{proof}

\color{black}
\section*{Appendix B}

\begin{proposition}(Existence of Globally Lipschitz Extension).  
\label{exl}
Let the system \eqref{sys} satisfy Assumption \ref{assmp} and let the safe set $\mathcal{S}$ be open and bounded. Then there exists a function $\bm{F}(\cdot,\cdot):\mathbb{R}^n \times \mathcal{D} \to \mathbb{R}^n$ such that 
\begin{enumerate}
    \item $\bm{F}(\bm{x},\bm{d}) = \bm{f}(\bm{x},\bm{d})$ for all $\bm{x} \in \overline{\mathcal{S}}$ and $\bm{d} \in \mathcal{D}$;
    \item $\bm{F}$ is globally Lipschitz in $\bm{x}$ uniformly over $\bm{d} \in \mathcal{D}$;
    \item $\bm{F}$ preserves the affine structure: $\bm{F}(\bm{x},\bm{d}) = \bm{F}_1(\bm{x}) + \bm{F}_2(\bm{x})\bm{d}$.
    \item $\bm{F}$ preserves the affine structure: $\bm{F}(\bm{x},\bm{d}) = \bm{F}_1(\bm{x}) + \bm{F}_2(\bm{x})\bm{d}$.
\end{enumerate}
\end{proposition}
\begin{proof}
We construct $\bm{F}$ through component-wise extension:

Let $\bm{f}_1(\bm{x}) = (f_{1,1}(\bm{x}), \ldots, f_{1,n}(\bm{x}))^\top$ and $\bm{f}_2(\bm{x}) = [f_{2,ij}(\bm{x})]_{n \times m}$. Since $\overline{\mathcal{S}}$ is compact and each component function is locally Lipschitz, they are Lipschitz continuous on $\overline{\mathcal{S}}$.

By Kirszbraun's Theorem, for each component $f_{1,i}(\cdot): \mathbb{R}^n \to \mathbb{R}$ ($i=1,\ldots,n$), there exists a globally Lipschitz extension $\tilde{f}_{1,i}(\cdot): \mathbb{R}^n \to \mathbb{R}$ with the same Lipschitz constant. Similarly, for each entry $f_{2,ij}(\cdot): \mathbb{R}^n \to \mathbb{R}$ of $\bm{f}_2$, there exists a globally Lipschitz extension $\tilde{f}_{2,ij}(\cdot): \mathbb{R}^n \to \mathbb{R}$.

Define $\bm{F}_1(\bm{x}) = (\tilde{f}_{1,1}(\bm{x}), \ldots, \tilde{f}_{1,n}(\bm{x}))^\top$ and $\bm{F}_2(\bm{x}) = [\tilde{f}_{2,ij}(\bm{x})]_{n \times m}$, and let

\[
\bm{F}(\bm{x},\bm{d}) = \bm{F}_1(\bm{x}) + \bm{F}_2(\bm{x})\bm{d}.
\]

This construction ensures:
\begin{enumerate}
    \item For $\bm{x} \in \overline{\mathcal{S}}$, $\bm{F}_1(\bm{x}) = \bm{f}_1(\bm{x})$ and $\bm{F}_2(\bm{x}) = \bm{f}_2(\bm{x})$, so $\bm{F}(\bm{x},\bm{d}) = \bm{f}(\bm{x},\bm{d})$;
    \item Since each component is globally Lipschitz and $\mathcal{D}$ is compact, $\bm{F}$ is globally Lipschitz in $\bm{x}$ uniformly over $\bm{d} \in \mathcal{D}$;
    \item The construction explicitly maintains the affine form in $\bm{d}$. 
\end{enumerate}

Therefore, $\bm{F}$ satisfies all required properties. 
\end{proof}

\begin{proposition}[Marginal Function Construction]
\label{thm:marginal-function}
Define the marginal function $B(\cdot,\cdot): [0,\infty) \times \mathbb{R}^n \to (-\infty, 0]$ as the negative distance between the boundary $\partial\mathcal{S}$ and the set of states reached by solutions to system \eqref{sys1} starting from $\bm{x}$ over $[t,T]$:
\[
B(t,\bm{x}) = -\inf\{\|\bm{\psi}_{\bm{x}}^{\bm{d}}(\tau)\|_{\partial\mathcal{S}} \mid \tau \in [t,T], \bm{d} \in \mathcal{M}\},
\]
where $\|\bm{\psi}_{\bm{x}}^{\bm{d}}(\tau)\|_{\partial\mathcal{S}} = \min_{\bm{y} \in \partial\mathcal{S}} \|\bm{\psi}_{\bm{x}}^{\bm{d}}(\tau) - \bm{y}\|$. If system \eqref{sys1} is safe over $[0,T]$, then there exists a time-dependent barrier certificate constructed from $B(t,\bm{x})$.
\end{proposition}
\begin{proof}
From Proposition 3 in \cite{maghenem2022converse}, $B$ is locally Lipschitz continuous. Following Theorem 2 in \cite{maghenem2022converse}, the map $t \mapsto B(t,\bm{\psi}_{\bm{x}}^{\bm{d}}(t))$ is monotonically nonincreasing. If system \eqref{sys1} is safe, then 
$\|\bm{\psi}_{\bm{x}}^{\bm{d}}(\tau)\|_{\partial\mathcal{S}}<0$ holds for all $\tau\in [0,T], \bm{x}\in \mathcal{X}_0$, and $\bm{d}\in \mathcal{D}$. Thus, by Lemma \ref{equiv1}, we can obtain that there exists $\delta > 0$ such that $B(0,\bm{x}) \leq -\delta$ for $\bm{x} \in \mathcal{X}_0$. Moreover, $B(t,\bm{x}) = 0$ for $t \in [0,T]$ and $\bm{x} \in \partial\mathcal{S}$, implying $B(t,\bm{x}) \geq 0$ for $t \in [0,T]$ and $\bm{x} \in \partial\mathcal{S}$. Following the proof methodology of Lemma \ref{converse_pre}, we can approximate $B$ by a continuously differentiable function that serves as a time-dependent barrier certificate.
\end{proof}

\end{document}